\documentclass[conference]{IEEEtran}
\IEEEoverridecommandlockouts

\usepackage{cite}
\usepackage{amsmath,amssymb,amsfonts}
\usepackage{algorithmic}
\usepackage{graphicx}
\usepackage{textcomp}
\usepackage{xcolor}
\usepackage[hyphens]{url}

\usepackage{myresearchpack}
\usepackage{enumitem}
\usepackage{tabularx}
\usepackage{multirow}
\usepackage{subcaption}
\usepackage{setspace}
\usepackage{thmtools, thm-restate}

\usepackage{amsthm}
\newtheorem{theorem}{Theorem}[section]
\newtheorem{definition}{Definition}[section]
\usepackage{amsmath}
\usepackage{listings}
\usepackage{color}

\def\BibTeX{{\rm B\kern-.05em{\sc i\kern-.025em b}\kern-.08em
    T\kern-.1667em\lower.7ex\hbox{E}\kern-.125emX}}
\begin{document}

\pdfpagewidth=8.5in
\pdfpageheight=11in

\newcommand{\iscasubmissionnumber}{306}

\pagenumbering{arabic}
\shortname{DataGuard}
\def\Xe{DataGuard\textsubscript{\texttt{ex}}\xspace}
\title{\X: Guaranteeing Private Training in Systolic-array Based Accelerators}
\author{
\IEEEauthorblockN{
Pawan Kumar Sanjaya\IEEEauthorrefmark{1},
Christina Giannoula\IEEEauthorrefmark{2},
Nikhil Shreekumar,
Ian Colbert\IEEEauthorrefmark{3},
Alec Dewulf\IEEEauthorrefmark{1},
Mehdi Saeedi\IEEEauthorrefmark{3}, \\
Ihab Amer,
Gabor Sines\IEEEauthorrefmark{3},
Nandita Vijaykumar\IEEEauthorrefmark{1}
}
\IEEEauthorblockA{\IEEEauthorrefmark{1}University of Toronto}
\IEEEauthorblockA{\IEEEauthorrefmark{2}Max Planck Institute for Software Systems}
\IEEEauthorblockA{\IEEEauthorrefmark{3}Advanced Micro Devices (AMD)}
\thanks{This work was partially conducted while Christina Giannoula and Nikhil Shreekumar were with the University of Toronto, and while Ihab Amer was with AMD.} 
\thanks{© 2026 Advanced Micro Devices, Inc. All rights reserved. AMD, the AMD Arrow logo, Radeon, and combinations thereof are trademarks of Advanced Micro Devices, Inc. Other product names used in this publication are for identification purposes only and may be trademarks of their respective companies.}
}

\maketitle
\thispagestyle{plain}
\pagestyle{plain}


\begin{abstract}

Differential privacy (DP) and federated learning (FL) have emerged as important privacy-preserving approaches when using sensitive data to train machine learning (ML) models. FL ensures that raw sensitive data does not leave the users' devices by training the model locally on the device. DP ensures that the model does not leak any information about an individual by \emph{clipping} and adding \emph{noise} to the gradients before updating the model. It provides formalism to constrain privacy loss during training to a \emph{privacy budget} determined \textit{a priori} by the owner of sensitive data. However, real-life deployments of FL algorithms typically assume that a third-party FL application can be trusted to correctly implement DP algorithms. Thus, the third-party application is given full access to sensitive data. 
In this work, we propose \X, a hardware-based mechanism that guarantees that the only data that can leave the device is the result of computation that meets DP requirements. \X can thus be used to ensure that the privacy budget defined by the data owner is not exceeded during FL training without the need to trust a third-party application. 
We evaluate \X in simulations of four accelerators for various ML models and demonstrate only small area overheads of less than 0.01\% and performance slowdowns of less than 0.3\%. 

\end{abstract}

\section{Introduction}
\label{introduction}

Protecting the privacy of sensitive user data in machine learning (ML) training is of critical importance in many domains such as healthcare and finance. For example, training medical diagnosis models requires sensitive patient records from multiple hospitals~\cite{cvd,multi,fedmed,med}. In the financial sector, sensitive user data from multiple institutions is used to train models for fraud detection~\cite{fedfinance} and credit scoring~\cite{li_dynamic_2024}. Models for text prediction~\cite{mkp}, facial recognition~\cite{fedface},
and others~\cite{fedinfo,fedhar,fedrecon,fedspeech,amfed} use sensitive user data collected on mobile phones. 

Federated learning (FL)~\cite{fl} and differential privacy (DP)~\cite{dp} protect user privacy when training ML models on sensitive data. In a typical FL system (Fig.~\ref{fig:fl}), a shared model is trained without the sensitive training data (\circled{a}) ever leaving the data owners' (\ie clients') devices (\circled{b}), \eg local servers or mobile devices. Instead, in each round, the training application (\circled{c}) trains a copy of the shared model on the client device and shares only the weight gradients with a central server (\circled{d}) for aggregation. By keeping data where it is owned, FL not only protects data privacy but also facilitates compliance with regulations such as GDPR~\cite{gdpr} that strictly restrict data transfers and storage locations.

\begin{figure}
    \centering
    \includegraphics[scale=0.75,trim={5mm 2.5mm 10mm 0.5mm}]{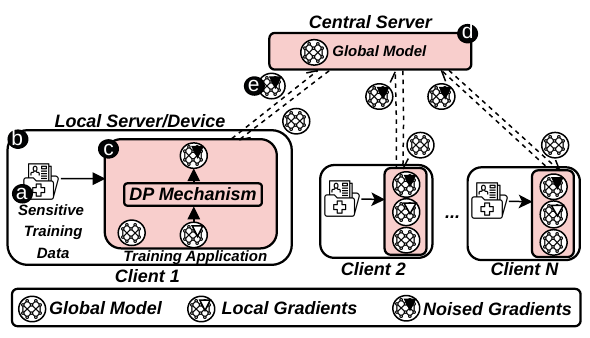}    
    \caption{Overview of DP-enabled Federated Learning.}
    \label{fig:fl}
    \vspace{-0.5cm}
\end{figure}

FL alone does not guarantee privacy as information about sensitive training data can be extracted from trained models~\cite{attack1,attack3}. DP~\cite{dp} can be used with FL to provide formal guarantees~\cite{dpsgd,dp2,llm-dp2}. DP is a mathematical framework to protect the privacy of an individual when their data is used in an analysis of a larger dataset. It has proven effective for statistical analyses~\cite{dpus,dpus2,rappor,apple2017differential} and more recently, for privacy-preserving ML on sensitive data~\cite{dpsgd,dp2,ldp2,ldp-fed,llm-dp1,llm-dp2,dpadapt,dpfyml,med,fedmed,fedfinance,shuffledp,bayesiandp}. The common approach integrates DP into SGD~\cite{dpsgd} by clipping gradients to a threshold and adding Gaussian noise before updating model weights. In DP-enabled FL, only noised gradients (\circled{e}) are sent to the central server (\circled{d}) for aggregation. Every time a user's data is used in a training iteration, the model learns some information about that individual, quantified as the \emph{privacy cost}. Data owners set a \emph{privacy budget}, and each iteration consumes some of it. Ensuring that the budget is not exceeded provides formal privacy guarantees.

\noindent
\textbf{Untrusted third-party training applications.} 
In real-world deployments, the FL training application is typically owned by a \emph{third-party} entity distinct from the clients who own sensitive data, as it requires aggregating gradients across independent clients (data owners). For example, in medical ML training, hospitals own local servers containing patient records; in facial recognition, users control mobile devices holding personal images. In both cases, the aggregator may be an ML or mobile application service provider. Consequently, this third-party application must be granted access to sensitive data on client devices, yet clients cannot easily verify that only correctly noised gradients leave the device rather than raw user data or a non-DP transformation of it. 

Therein lies the challenge: existing ML frameworks that incorporate FL and DP require trusting third-party applications with access to sensitive data. Verifying that a training application correctly implements the DP algorithm is non-trivial, as the application operates directly on raw sensitive data and may violate DP through subtle deviations from the algorithm. Such deviations can stem from malicious design or programmer error and are difficult to detect.

Software-based approaches have significant disadvantages and shortcomings as we discuss in~\S\ref{sec:motivation}. Program analysis~\cite{solo,duet,lightdp,fuzzi,distance,hoare} can verify that an application satisfies DP, but requires assuming an honest programmer (\ie trusting the third party) and is thus unsuitable in our setting. Dynamic taint tracking~\cite{taintart,taintdroid} and software DP frameworks~\cite{pinq,gupt} can enforce privacy at runtime without trusting the programmer, but incur prohibitive performance overheads on ML accelerators due to tracking data in memory or enforcing isolation. Code attestation~\cite{takagi2025securingprivatefederatedlearning} avoids runtime overheads by having a trusted authority verify and sign the application before deployment; however, such an approach incurs several practical challenges in FL, such as requiring verification and signing for every code modification (see \S\ref{sec:motivation}). 

To obviate the need to trust a third-party application, prior works~\cite{dpfyml,ldp_external,dwork_calibrating,randomized_response,simha_cal,ulp_ldp,serandip} propose noising sensor data \emph{before} granting the training application access to it. Since DP is robust to post-processing, any subsequent computation on the noised data does not affect privacy. However, this significantly degrades model accuracy compared to noising gradients during training~\cite{dpfyml}, and has not been applied to text, images, or videos. Thus, providing DP guarantees against an untrusted third-party training application remains an open problem.

\noindent
\textbf{Goal.} In this work, we explore the feasibility of \emph{leveraging lightweight hardware} to provide DP guarantees and serve as a substrate within the accelerator for privacy-preserving ML training. To this end, we propose \X, the first framework for ML accelerators, to our knowledge, that enforces DP guarantees using hardware mechanisms. 
This obviates the need to trust third-party software in DP-based ML training. To enable this, the framework must provide the following guarantees: 
1) there is no leakage of sensitive training data, \ie only noised gradients can be shared out of the device;
2) gradients that are shared out of the device are noised according to the DP algorithm; and 
3) privacy cost of training should not exceed the defined budget.

\noindent
\textbf{Challenges} in designing a hardware-based DP framework:
\textit{Challenge 1:} Preserving the programmability of the ML accelerator. 
A naïve design would hard-code the full DP-SGD algorithm directly into accelerator hardware. However, this would require specialized datapaths and on-chip memories for all computation in order to track and isolate the computation of gradients and their noising. This would require extensive redesign of modern accelerators, which rely on shared compute and memory structures. Such specialization would severely restrict programmability, complicate support for different DP mechanisms (e.g., layer-wise vs. global clipping~\cite{layer_clip}), and make enforcing privacy budgets challenging.

\noindent
\textit{Challenge 2:} Efficiently tracking gradients in memory. Given that the DP mechanisms cannot be baked into the datapath, the system must \emph{dynamically and automatically identify} the locations of calculated raw and noised gradients in memory to ensure that only correctly noised gradients are shared out of the device. It is challenging to track the gradients in memory during program execution without compromising performance or incurring substantial area overheads. 

\noindent
\textbf{\X}. Our proposed framework consists of two hardware components. First, a lightweight tagging mechanism identifies and tracks \emph{correctly noised and clipped} gradients in hardware. Only data tagged as such by \X is allowed to leave the device by a privileged software manager. Second, a noising module ensures that gradients are \emph{correctly noised and clipped} to satisfy DP guarantees. An untrusted application is \emph{forced} to use this module, because only gradients noised by the module are tagged as ``safe'' by \X. The module is accessed using new instructions that mirror existing accelerator instructions, thereby preserving programmability. With \X, the application can flexibly perform \emph{any computation} on the ML accelerator using any data; however, only data that meets the required (and specified) differential privacy guarantees and budget will be marked ``safe''.

We also propose an extension, \Xe, that supports additional guarantees with slightly higher area and performance overheads: the training samples are randomly sampled every iteration, and intermediate data cannot be accessed between iterations. \Xe enables using DP accounting methods with tighter bounds, where more training iterations can be performed for the same privacy budget. \X and \Xe together support all modern DP-based FL algorithms that clip and add Gaussian noise~\cite{nbafl,fedavg,dpadapt,fl_dp,llm-dp1,llm-dp2,yang2023dynamic,cheng2022differentially} directly on the user device.

We model \X in SCALE-sim~\cite{scale_sim} simulator for TPU~\cite{tpuv2}, DIVA~\cite{diva}, and an accelerator based on Shidiannao~\cite{shidia}. Our evaluations on popular ML models demonstrate only small slowdowns ($<$0.3\%). We also implement \X in RTL and demonstrate negligible area (0.01\%) and power overheads (0.07\%). 
\noindent
We make the following contributions:
\setlist[itemize]{leftmargin=*}
\begin{itemize}[topsep=0pt]
    \item We motivate the need for additional privacy protection mechanisms in real-life deployments of DP-enabled FL. We explore the potential of leveraging hardware mechanisms to provide stronger DP guarantees.
    \item We identify key guarantees that must be met in hardware to ensure that user-defined privacy budgets cannot be violated by any application. We propose the first framework that provides DP guarantees in hardware for ML accelerators and obviates the need to verify third-party applications for ML training (formal proof in \S\ref{discussion}).  
    \item We propose two variants of \X to support all DP-based FL algorithms for local DP settings \ie the most secure setting where only data that satisfies DP guarantees can be shared out of the device.  
    \item We evaluate \X in simulation for various ML training workloads on four ML accelerators and show only small performance overheads. We implement \X in RTL and show negligible area and power overheads.  
\end{itemize}

\brokenpenalty=0
\widowpenalty=0
\section{Background}
\label{background}
\subsection{Differential Privacy (DP)}
\label{dp_background}

Differential privacy (DP)~\cite{dp} is a theoretical framework that bounds the influence any one individual data point can have on the outputs of a statistical function. A function over a dataset with user information 
satisfies DP if and only if the output of the function does not change significantly when a single data point (user or record) is added or removed from the input dataset. Here, DP assumes that the attacker can only observe the outputs of the statistical function. Therefore, if the function's output does not change meaningfully depending on whether a particular data point is included in the dataset, the attacker cannot reliably determine whether that record was part of the input dataset. This is enough to ensure meaningful protection against adversaries, even if they may have arbitrary auxiliary information about the user~\cite{dp}. 
Formally, for any two adjacent datasets, DP restricts the difference in probability distributions of the function's outputs, as stated in the following definition.
 
\begin{spacing}{0.9}
{\vspace{-0.5em}
\begin{definition}
\textbf{Differential Privacy:} Any function $\mathcal{M}: X\xrightarrow{}Y$, is  $(\epsilon,\delta)$ \textit{differentially private} if for two adjacent datasets $D$ and $D'$ \ie $\left|D\right|-\left|D'\right| = 1$ and $\forall K \in Y$, satisfies
\vspace{-0.5em}
\begin{equation}
   P(\mathcal{M}(D) = K) \leq e^\epsilon P(\mathcal{M}(D') = K) + \delta
   \label{equ_eddp}
\end{equation}
\vspace{-0.7cm}
\label{def:adjacency}
\end{definition}}
\end{spacing}
Here, ($\epsilon,\delta$) denotes the \textit{privacy cost}. When $\epsilon$ is sufficiently small, an adversary is unable to distinguish between adjacent datasets based on $\mathcal{M}$'s output. $\delta$ can be interpreted as the probability that the entire dataset is leaked
. Typically for a dataset with $n$ records, $\delta \ll \frac{1}{n}$ is recommended~\cite{near_abuah_2021}.

To ensure that the final output of any function ($F$) satisfies $(\epsilon, \delta)$-DP, it can be post-processed by adding noise sampled from a Gaussian distribution ($\mathcal{N}$) (\ref{equ_noising})~\cite{gauss_dp}. 
\vspace{-0.5em}
\begin{equation}
   M(x) = F(x) + \mathcal{N}\left(0,\sigma^2\right), \sigma = \sqrt{2S^2\log{\frac{1.25}{\delta}}}/\epsilon 
   \label{equ_noising}
   \vspace{-0.5em}
\end{equation}

Here, $S$ is the $\ell_2$-sensitivity of the function. The $\ell_2$-sensitivity is the maximum value by which the output of the function changes when a single entry in the dataset is removed. 

DP is robust to post-processing, \ie any subsequent analysis on the output of a DP function also satisfies $(\epsilon,\delta)$-DP. DP also  supports composition: if two functions that satisfy $(\epsilon{}_1,\delta{}_1)$-DP and $(\epsilon{}_2,\delta{}_2)$-DP are chained together to construct a third function, \ie $f_3(x) = f_2(f_1(x))$, the resulting function ($f_3$) satisfies $(\epsilon{}_1+\epsilon{}_2,\delta{}_1+\delta{}_2)$-DP. Thus, using composition, the privacy cost of an iterative algorithm such as ML training can be calculated using the privacy cost of a single iteration~\cite{dp}. In FL-training, a single training \emph{round} can consist of multiple training \emph{iterations} where the clients update the local model weights. If each iteration satisfies $(\epsilon,\delta)$-DP, the privacy cost of a round can be calculated using composition.
\vspace{-0.5em}
\subsection{Differential Privacy in Federated Learning}
\label{back_dp}
The DP formalism can be applied to stochastic gradient descent (SGD)-based training to provide DP guarantees when using FL to train models on sensitive data. It can be modified in two ways: \textit{per-batch} and \textit{per-example} clipping~\cite{fl_dp}. In both, noise is sampled from a Gaussian distribution with standard deviation $\sigma=C_{th}^2 z^2$, where $z$ is a user-defined scaling factor derived using $(\epsilon,\delta)$ and $C_{th}$ is a user-defined threshold for the $\ell_2$-norm of all the gradients (${\lVert g\rVert_2}$). The $\ell_2$-norm of the gradients (${\lVert g\rVert_2}$) is the square root of sum of squares of all the gradients calculated per iteration. Before noising, it is ensured that the $\ell_2$-norm of the gradients is less than $C_{th}$ by \emph{clipping} which involves scaling each gradient by $\min\left(1,\frac{C_{th}}{\norm{g}_2}\right)$.

\noindent
\textbf{Per-batch clipping:} For every training iteration, the gradients for the batch ($g$) are calculated (where gradients across examples in a batch are first added) and clipped such that the $\ell_2$-norm is less than $C_{th}$ (\ie  ${\norm{g}_2} \leq C_{th}$). Noise is added to these clipped gradients to obtain the gradients for the batch ($\hat{g}_n$) which are then used to update the model weights. 
\begin{spacing}{0.8}
{
\vspace{-0.5em}
\begin{equation}
\begin{aligned}  
    \hat{g}_n = g_c + \mathcal{N}\left(0,4C_{th}^2z^2I \right),\\
    \textrm{where } g_c = \min\left(1,\frac{C_{th}}{\norm{g}_2}\right)\cdot g 
    \label{equ_ldp}
\end{aligned}
\end{equation}
}\end{spacing} 
\noindent
\textbf{Per-example clipping:} In every training iteration, instead of calculating the gradients for the entire batch (of batch size $b$), per-example gradients ($g_i$) are calculated and the gradients are clipped such that the $\ell_2$-norm of the gradients for each example is less than $C_{th}$. Noise is then added to the summed gradients for the batch ($g_c$) and model weights are updated using the noised gradients ($\hat{g}_n$), as follows:
\begin{spacing}{0.8}
{
\vspace{-0.5em}
\begin{equation}
\begin{aligned}   
    \hat{g}_n = g_c + \mathcal{N}\left(0,C_{th}^2z^2I \right),\\
    \textrm{where } g_c = \sum_{i=1}^{b} \min\left(1,\frac{C_{th}}{\norm{g_i}_2}\right)\cdot g_i 
    \label{equ_dpsgd}
\end{aligned}
\end{equation}
}\end{spacing}

To enable DP-based ML training, the gradient computation for the entire batch ($g_c$) in every iteration should satisfy DP. In the \textit{per-example clipping}, the per-example gradients ($g_i$) are clipped such that $\norm{g_i}_2 \leq C_{th}$. In doing so, $\ell_2$-sensitivity of the computation is fixed to $C_{th}$. With \textit{per-batch clipping}, the batch gradient is clipped to $\lVert g\rVert_2 \le C_{th}$, and thus the sensitivity of the computation is $2C_{th}$. Based on (\ref{equ_noising}), the gradient computation during training satisfies $(\epsilon, \delta)$-DP, if (\ref{equ_dp_noise}) is true and by post-processing property the model also satisfies \mbox{$(\epsilon,\delta)$-DP}~\cite{dp}.
\begin{spacing}{0.8}
{
\vspace{-0.5em}
\begin{equation} \begin{aligned} z = \sqrt{\left(2\log{\frac{1.25}{\delta}}\right)}/\epsilon~, \label{equ_dp_noise} \end{aligned} \end{equation} }
\end{spacing}

Per-batch clipping ensures that an attacker cannot infer if a particular \emph{client} contributed to the training dataset, \eg used in LDP-FED\cite{ldp-fed} and NbAFL~\cite{nbafl}. Per-example clipping ensures that an attacker cannot infer if a particular \emph{record} was present in the training dataset, \eg used in federated variants of the DP-based SGD algorithm~\cite{dpsgd,fl}. A model can be trained to a similar accuracy using both clipping mechanisms. However, per-batch clipping requires more communication rounds, while per-example clipping requires larger clipping thresholds~\cite{fl_dp}.



\noindent
\textbf{Privacy Budget:} Information leaked during gradient computation is bounded by $(\epsilon,\delta)$, \ie the privacy cost. Each computation on the dataset increases this cost~\cite{dp2}. For example, consider one training round has a cost of $(\epsilon_1,\delta_1)$. Since training is iterative \ie multiple training rounds are performed over the same dataset; $T$ rounds cost $(T\epsilon_1,\,T\delta_1)$ by DP composition~\cite{dp}. Real-world deployments of such applications, an upper limit ($\epsilon_{max},\delta_{max}$) called \emph{privacy budget} is defined to ensure that the privacy cost does not exceed a threshold, \ie $T\epsilon_1\le\epsilon_{\max}$ and $T\delta_1\le\delta_{\max}$. Thus, each $(\epsilon,\delta)$-DP training round can be seen as \emph{consuming} a part of this \emph{budget}, and training is stopped once the \emph{budget} is exhausted. 

To reduce the cost of multiple rounds of training, modern DP-ML algorithms which perform \textit{per-example clipping} use \textit{privacy amplification by subsampling}~\cite{dpsgd}: records used for each training iteration are randomly sampled with replacement from the dataset. Using a randomly sampled batch reveals less information about each record, and reduces the cost of $T$ rounds to $(q\sqrt{T}\,\epsilon,\,\delta)$, where $q<1$ is the sampling probability.

\vspace{-0.25cm}

\section{Threat Model and Assumptions}
\label{assumptions}
This work targets the common FL scenario in which a third-party entity trains a model on
data from independent clients, such as a medical diagnosis model trained across multiple
hospitals~\cite{fedmed} or a fraud detection model built on transaction data from financial
institutions~\cite{fedfinance}. We model the FL system with two types of entities: 1) \emph{Data owners} are the independent entities that own the sensitive data and the devices
or servers used for local model training; each data owner also controls its own privacy
budget allocation during training. 2) \emph{Aggregator} is the entity that trains a global model using sensitive data from multiple independent clients, using its own server for aggregating noised gradients.

Our threat model assumes the aggregator is an untrusted and potentially malicious entity from the data owners' perspective. The aggregator controls both the training application executed on client devices and the central server where gradients are aggregated. While conventional OS-based privacy mechanisms protect sensitive data through access control, they offer no protection once an application is granted data access. The training application owned by the aggregator can therefore share sensitive data, non-DP gradients, or any other non-DP transformations of sensitive data outside the client device. The aggregator can thus perform a range of privacy-violating adversarial attacks on the aggregated model, including inference attacks to learn private information about the clients' datasets~\cite{attack1,attack3}. Since the data owners control the client devices, we assume the local training infrastructure (including hardware and system software) is trusted, with security vulnerabilities falling outside the scope. We reasonably assume the aggregator has no physical access to these systems, \ie physical side-channel attacks cannot be used to extract sensitive information.

DP-based FL mechanisms fall into two categories depending on where the DP algorithm is implemented: Centralized DP (CDP) and Local DP (LDP). In CDP, the DP mechanism is implemented at the central server owned by the aggregator; in LDP, it runs on the client devices owned by the data owners. CDP requires transmitting raw, unnoised gradients outside the local trust boundary and, therefore, inherently requires trusting the aggregator. Thus, CDP cannot be used under our threat model, where the third-party aggregator is untrusted, and we focus exclusively on LDP.  


\vspace{-5pt}
\section{A Case for Hardware-based Guarantees}
\label{sec:motivation}


To motivate the benefits of providing primitives in hardware for DP guarantees, we first discuss a purely software-based solution that could potentially address the challenge of trusting third-party applications. We note that many approaches below are hypothetical as they have not been proposed or constructed to our knowledge to address the above challenge. 

\emph{(i) Static or dynamic program analysis~\cite{solo,duet,lightdp,fuzzi,distance,hoare,dfuzz,chorus,ektelo,dduo}.}
These works require programmers to annotate variables that store sensitive data and specify privacy rules governing data usage, effectively assuming an ``honest but fallible'' programmer. They cannot be used in settings where the programmer, \ie the third party owning
the training application, is untrusted.


\emph{(ii) Dynamic taint tracking-based approaches~\cite{taintart,taintdroid} in software}  can enforce privacy properties at runtime. However, they rely on a trusted privileged software layer to instrument and propagate taint. Thus, they are not directly applicable to ML accelerators such as TPUs~\cite{tpu}, which lack support for privilege-based isolation mechanisms. They also incur substantial performance overheads (upto 15\%~\cite{taintart}).

\emph{(iii) Code attestation-based approaches~\cite{takagi2025securingprivatefederatedlearning}},
where a trusted authority verifies that the FL training code meets DP requirements, compiles
it, and generates a cryptographic signature. Clients validate this signature during
training to confirm that the application was compiled by the trusted party. This approach has
several limitations: (a) it requires a separate trusted party distinct from the clients
and aggregator; (b) verifying that code satisfies DP requirements is non-trivial, as it requires manual analysis by experts when the programmer is untrusted; (c) the application must be compiled and signed per platform, which is
challenging as clients commonly use heterogeneous devices and operating systems in FL; and (d) any
update to the application code requires re-verification, recompilation,
and re-signing.

\emph{(iv) Software frameworks~\cite{pinq,gupt,psi} that restrict access to sensitive data.} These frameworks expose a controlled API through which the application executes custom functions on sensitive data. The outputs are automatically transformed to satisfy DP guarantees before being returned to the application. These approaches suffer from two significant limitations. First, they impose substantial performance overhead on ML accelerators. To enforce DP guarantees, such frameworks must isolate each execution of the third-party training function to prevent reuse of intermediate results across training examples and iterations. OS-level memory isolation can enforce this on the host CPU, but ML accelerators lack hardware support for such mechanisms. Device memory is managed entirely by the application in software. To prevent unintentional or malicious data reuse, the framework \emph{must} zero out accelerator memory between invocations, an operation that takes 0.3\,s on a TPUv3~\cite{tpuv2} (32\,GB memory, 900\,GB/s bandwidth) per training iteration at batch size~8. Our evaluations show that the same iteration for VGG16 completes in $\sim$0.03\,s, a $10\times$ slowdown. Second, several DP-based FL applications use local personalization~\cite{feddpa,fedsam} to improve training accuracy. They store a local non-DP state that is used to modify the gradients before any DP transformation. Such frameworks do not support local personalization algorithms, as they prohibit the application from retaining any non-DP state between training rounds.

\emph{(v) Software ML frameworks with restricted APIs.} In this approach, the client device installs a \emph{validated} (trusted) ML framework that implements all training code. The third-party training application uses an API or messaging interface to provide the model description, weights, and training parameters (\eg learning rate, batch size). The framework trains the model and returns noised gradients to the application, which are shared with the server. This approach has two major limitations. First, ML frameworks are large codebases with complex interactions among components such as layer definitions, gradient computations, and noising. Validating that such a codebase correctly enforces DP requires manual expert analysis across all possible code paths and/or reliance on programmer annotations. Even validated codebases can subtly violate DP~\cite{dp_bug,tramer2022debuggingdifferentialprivacycase}, because DP violations need not manifest as correctness bugs. A bug in compilation tools or imported libraries that alters, for instance, the rounding mode used during clipping does not significantly affect training accuracy but can break the DP guarantee~\cite{dp_bug}. Validating such a framework, therefore, requires ensuring that the high-level algorithm, each operator's implementation, \emph{and} their interactions during training all preserve DP, a task whose difficulty scales with framework complexity. Second, any modification to the training algorithm, such as adopting adaptive clipping~\cite{dpadapt} or a new personalisation strategy~\cite{feddpa}, requires: (a) updating the independent training framework; (b) revalidating the updated codebase; and (c) redeploying across all participating clients. In a cross-device FL setting with heterogeneous clients, this requires coordination among multiple independent parties before the training application can adopt the training algorithm during deployment.



Thus, providing DP guarantees for FL in software requires combining several techniques: program analysis/verification and code attestation across the software stack, including various libraries and drivers. This approach is potentially error-prone~\cite{dp_bug,tramer2022debuggingdifferentialprivacycase}, requires coordinating across multiple parties, and limits programmability and ease of updates.

In this work, we aim to investigate the design of hardware-based checks for core DP invariants (\eg clipping and noising) that can serve as a \emph{substrate} for building privacy-preserving ML systems. These checks are exposed as concise interfaces for DP operations that the training application uses to compute gradients. This design supports modifications to the training algorithm without requiring re-verification of the application, effectively decoupling programmability from privacy enforcement. Once verified, these hardware checks ensure that any software stack running on the accelerator meets privacy guarantees.


\section{\X: Design Overview}
\label{approach}
\begin{figure}
\begin{center}
\begin{adjustbox}{%
    max totalsize={\linewidth}{\textheight}
}
    \centering
    \includegraphics[scale=1.5,trim={0.75cm 0.2cm 0.2cm 0.2cm},clip]{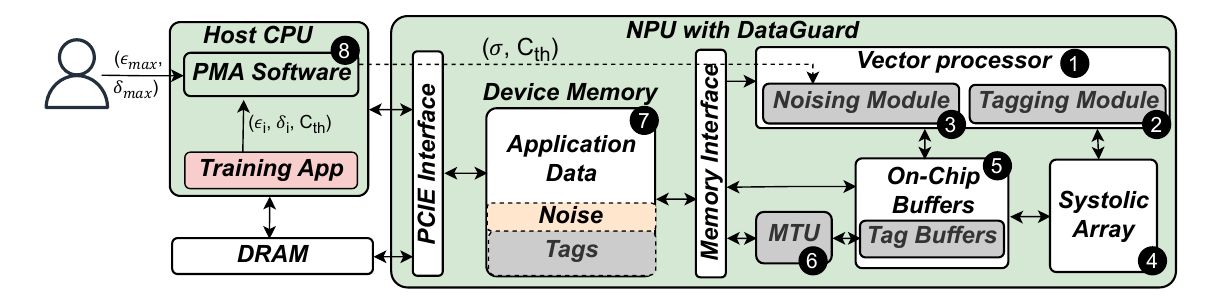}
\end{adjustbox}
    \vspace{-0.5cm}
    \caption{Overview of a system with \X.}
        \vspace{-1cm}
    \label{fig:overview}
\end{center}
\end{figure}

We introduce \X, a framework that provides DP guarantees in hardware for ML accelerators such as TPUs~\cite{tpu,tpuv2,tpuv4}, DIVA~\cite{diva} and Genesys~\cite{genesys}. 
We assume the ML accelerator is a separate device with dedicated device memory and all operations can be performed on the accelerator. The host CPU schedules the necessary computations on the accelerator and coordinates data movement. Fig. \ref{fig:overview} presents an overview of a system with \X, having an example accelerator with a systolic array (\circled{4}), a dedicated memory (\circled{7}), and a vector processor (\circled{1}). Modules shaded in green indicate the trusted modules; we trust the host CPU and the privileged software running on the host. 
We also trust the accelerator and hardware modules of \X. \X interfaces with any third-party training application via a privileged software layer: the privacy management application (PMA) (\circled{8}).
The data owner sets the privacy budget $(\epsilon_{\max},\delta_{\max})$. Using Eq.~\ref{equ_ldp}, the PMA computes the standard deviation for noising ($\sigma$) and the maximum training rounds ($N$) allowed within this budget. The PMA cannot itself guarantee that a third-party application adheres to $(C_{th},\epsilon_i,\delta_i,\epsilon_{\max},\delta_{\max})$; it configures the hardware with these parameters, and \X enforces them in hardware.


To provide guarantees that each training round is ($\epsilon_{i}$,$\delta_{i}$)-DP and the ($\epsilon_{max}$,$\delta_{max}$) budget specified by the user is never exceeded, \X guarantees that the following conditions are satisfied: 1) \textbf{Noising condition:} Ensuring 
that the data is noised by sampling from the correct distribution, \ie with the $\sigma$, decided by the PMA.  2) \textbf{Clipping condition:} Ensuring 
that the $\ell_2$-norm of the noised data is less than or equal to $C_{th}$, as configured by the PMA. The application can perform the gradient clipping using either per-example or per-batch clipping.  3) \textbf{Budget condition:} The training algorithm does not exceed the ($\epsilon_{max},\delta_{max}$) as provided by the user. The clipping and noising conditions together ensure that the ($\epsilon_i,\delta_i$)-DP guarantees are satisfied for each round. However, the application may perform multiple training rounds and the total privacy cost is calculated using composition theorems (\S\ref{background}). Thus, to ensure that the total budget is not exceeded ($\epsilon_{max},\delta_{max}$), the number of training rounds must not exceed $N$, as calculated by the PMA. 4) \textbf{Privacy condition:} Only the noised gradients are allowed to leave the device. 
\vspace{-4pt}
\subsection{\X: Key Ideas}
Fig. \ref{fig:key_mod} shows the two operations implemented by \X to provide the necessary guarantees: 1) $noise\_l2$ (\circled{a}) and 2) $privacy\ budgeting$ (\circled{b}). Regardless of the computation performed by the application (\circled{c}), any data that leaves the device must first be transformed using the noising operation. The $noise\_l2$ operation ensures that the data is noised with Gaussian noise with a specified $\sigma$. The $\ell_2$-norm of all the data that is noised is also computed during the $noise\_l2$ operation. The $privacy\ budgeting$ operation then performs the necessary checks to ensure that the overall privacy budget ($\epsilon_{max}, \delta_{max}$) is not exceeded. We now explain how these operations are used to satisfy the necessary conditions enumerated previously. We explain how these are implemented in \S\ref{detailed_design} and perform a security analysis in \S\ref{sec_analysis}.

\begin{figure}[h]
\begin{center}
\vspace{-0.4cm}
\begin{adjustbox}{%
    max totalsize={\linewidth}{\textheight}
}
    \centering
    \includegraphics[trim={0.2cm 0cm 0 0.1cm},clip,scale=0.8]{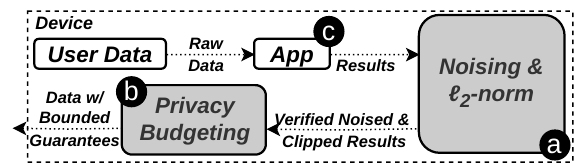}
\end{adjustbox}
\vspace{-0.25cm}
    \caption{Key operations of \X.}
    \vspace{-0.5cm}
    \label{fig:key_mod}
\end{center}
\end{figure}

\subsubsection{Verifying the privacy condition} 
All the data in device memory (\circled{7}) is associated with tags. These tags are used to identify data that is ``safe", \ie data that has been transformed using the operations \circled{a} and \circled{b} in Fig. \ref{fig:key_mod}. 
To do so, we implement a lightweight information flow tracking (IFT) and tagging mechanism in the tagging module (\circled{2}).

\subsubsection{Verifying the noising condition}
\label{noising_cond}
We implement a configurable noising module (\circled{3}) in hardware that implements the $noise\_l2$ operation (\circled{a}). This module is configured by the PMA to sample noise from a Gaussian distribution with a standard deviation of $\sigma$.  The training application uses a custom instruction (\texttt{add-noise}) to perform noising using this module. Only the results generated by this module are tagged as ``noised" by \X. Thus, for any data to be shared out of the device it must be noised using the \texttt{add-noise} instruction. If the ``noised" data satisfies all the necessary conditions (outlined in \S \ref{clip_cond}, \S \ref{budget_cond}), it is ``safe" and is allowed to be shared out of the device.

This ensures that, for the noised gradients to be tagged correctly and communicated out of the device, it should be processed by the noising module of \X. However, an application can use the \texttt{add-noise} instruction arbitrarily in the code and thus put any initial, intermediate, or computed data through the noising module. This data would still be marked as ``noised" and may be shared out of the device. \X is designed such that the DP guarantees are met regardless of this (explained in \S \ref{clip_cond}, \S \ref{budget_cond}). 
\subsubsection{Verifying the clipping condition}
\label{clip_cond}
The noising module, in addition to adding noise, calculates and stores in hardware a cumulative sum of squares ($S$) for all data being noised. When a training round is completed and the application calls the \texttt{audit} instruction, \X checks that the \textit{clipping condition} is satisfied using the cumulative sum \ie $\sqrt{S} \leq C_{th}$ and resets it to $0$. Note that the \texttt{audit} instruction can be called arbitrarily by the application. This does not violate ($\epsilon_i$, $\delta_{i}$)-DP. However, the application must call both \texttt{add-noise} and \texttt{audit} instructions at the end of every training round to implement DP-enabled ML training correctly (to ensure accuracy of the trained model). These two instructions together implement the $noise\_l2$ operation (\circled{a}).

\subsubsection{Verifying the budget condition}
\label{budget_cond}
\X tracks the number of invocations of the \texttt{audit} instruction with a hardware counter. This counter stores the current \emph{epoch}, \ie \texttt{audit} calls made till the current execution point. 
Thus, when the \texttt{audit} instruction is called at the end of each training round, this counter is incremented. The epoch count is used by the PMA to calculate the total privacy cost of training (as explained in \S \ref{back_dp}) and verify that it does not exceed the privacy budget. For example, when the \texttt{audit} instruction is called $n$ times, the cost is calculated as ($n\epsilon, n\delta$). The application is allowed to share data out of the device only if the calculated privacy cost is less than the budget \ie $n\epsilon < \epsilon_{max}$ and $n\delta < \delta_{max}$. The application can use the \texttt{add-noise} and \texttt{audit} instructions for any data such as raw training samples, gradients, or any intermediate results. The privacy cost calculation in \X ensures that the budget is never exceeded, thereby protecting privacy. When used as intended for DP-enabled ML training, we assume that the application shares the noised gradients every time the \texttt{audit} instruction is called. If the application does not share data every time, \X overestimates the privacy cost. However, user privacy is preserved because the budget cannot be exceeded.

Together, \texttt{add-noise} and \texttt{audit}, provide the necessary guarantees: \texttt{add-noise} ensures that only properly noised values can ever be marked ``noised" and computes the sum of squares ($S$) of the data required to verify clipping, while \texttt{audit} verifies clipping (by using $S$ to compute $\ell_2$-norm) and budget conditions before the data is marked "safe" to leave the device. The application is thus \emph{forced} to use these instructions to transmit noised gradients to the server. Any data marked "safe" satisfy DP guarantees (\S\ref{discussion}).

\section{Detailed Design}
\label{detailed_design}
As shown in Fig. \ref{fig:overview} the ML accelerator 
consists of a systolic array (\circled{4}) with 64k processing elements for multiply-and-accumulate operations and a 128-lane wide vector processor~(VPU)~(\circled{1}) for other operations. 
Fig. \ref{fig:proc_changes} illustrates the modifications in the VPU. 
The accelerator performs training using single-precision floating-point numbers~\cite{ieee754_2019} \ie each element is 4 bytes. 
The noising~(\circled{2}) and the tagging~(\circled{2}) modules are implemented in the VPU. The on-chip buffers~(\circled{5}) are modified to store the tags and a memory tagging unit (MTU)~(\circled{6}) is implemented to perform the memory operations for the tags. The device memory~(\circled{7}) is partitioned to store the application data, tags, and sampled noise values. 
\begin{figure}[!htb]
\begin{center}
    \vspace{-0.2cm}
\begin{adjustbox}{%
    max totalsize={\linewidth}{\textheight}
}
    \centering
    \includegraphics[scale=0.75]{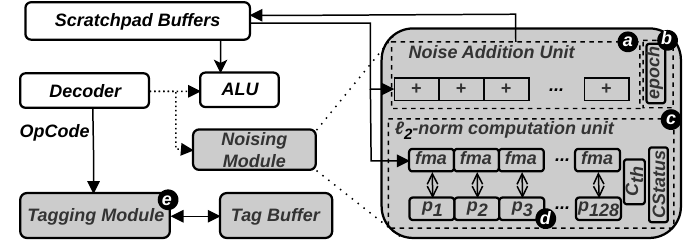}
\end{adjustbox}
    \vspace{-0.2cm}
    \caption{Block diagram of the vector processor in the accelerator. Modules added by \X are shaded.}
    \vspace{-0.85cm}
    \label{fig:proc_changes}
\end{center}
\end{figure}

\subsection{Noising Module}
\label{op:noising}
 Fig. \ref{fig:proc_changes} illustrates the components of noising module: noise addition unit~(\circled{a}), budgeting unit~(\circled{b}), and $\ell_2$-norm calculation unit~(\circled{c}). The noising module implements two instructions: \texttt{add-noise} and \texttt{audit}. 
 The \texttt{add-noise} instruction adds a floating-point (FP) operand and a noise term, then updates the $\ell_2$ norm using the operand. It requires twice as many FMA operations as a plain FP addition. Repurposing existing FMA units would degrade performance, so we implement dedicated units in the execute stage for both operations and use a pipelined implementation based on HardFloat~\cite{refJ} to avoid impacting the critical path.

\noindent
\textit{Noise addition unit:} The noise addition unit~(\circled{a}), is responsible for performing the noise addition during \texttt{add-noise} operation. The sampled values from the noise distribution are stored in a protected memory region on the device memory and are populated by the PMA (using a software implementation~\cite{harris2020array} of noise sampling) on the host CPU. This protected region is configured by the PMA using a configuration register (\emph{noise-br}) on the accelerator. The application is not allowed to access this memory region. The on-chip buffers are partitioned such that one-fourth of the buffers are used to store the sampled noise values during the noising operation. When the instruction is executed for the first time, a batched data transfer of 6 MB (size of on-chip buffers is 24 MB) is initiated from the device memory to on-chip buffers. These values are then added to the gradients to complete the noising operation. This is similar to how noising is implemented in other accelerators~\cite{diva}. 

The privacy guarantees when using DP depend on the correctness and freshness of the injected noise. The trusted PMA generates the noise values and writes them directly into the protected accelerator memory. We assume that the PMA samples from the correct distribution and produces fresh randomness for each training round. Importantly, \X remains completely agnostic to the concrete noise-generation method. Noise-generation mechanisms based on software sampling, hardware random number generators, or verifiable hardware-based approaches (\eg DINAR~\cite{dinar}) can be trivially integrated with \X as black-box components.

\noindent
\textit{Budgeting unit:} Consists of an 8-bit counter (\emph{epoch} \circled{b}) which is incremented every time an \texttt{audit} instruction is executed. It is reset to $1$ by the PMA before execution of the application. 

\noindent
\textit{$\ell_2$-norm calculation unit:} The $\ell_2$-norm calculation unit~(\circled{c}) maintains the cumulative sum of squares of the operands to \texttt{add-noise} instruction. The unit has 128 registers~(\circled{d})~($P_i$) that store partial sums of the $\ell_2$-norm and 128 fused multiply-add units, one for each vector lane in the VPU. The partial registers for each lane are updated with the sum of the squared operand for the corresponding lane and the initial value. 

The $\ell_2$-norm calculation unit also includes two registers: a status register (\emph{CStatus}) and a configuration register ($C_{th}$). The PMA configures the $C_{th}$ register with the clipping threshold before executing the application. When an \texttt{audit} instruction is executed, the partial sums in the registers~(\circled{d}) are added together to obtain the $\ell_2$-norm ($P_{agg}$) of all elements noised in the current iteration. The noising module then checks if the calculated $\ell_2$-norm is less than the threshold in $C_{th}$ \ie $\sqrt{P_{agg}} \leq C_{th}$. The \emph{CStatus} register is written with the value of the \emph{epoch} if the register was cleared before (\ie the current value is $0$) and the check fails. After the check, the partial sum registers~(\circled{d}) are cleared. The \emph{CStatus} register can only be cleared by the PMA and we explain how it is used in \S \ref{tagging}.

\subsection{Tagging Module}
\label{tagmod}
\X implements tagging using a lightweight mechanism. The data in the device memory is divided into blocks of 512 bytes (128 elements x 4 bytes), and each block has 1 byte of metadata, called tags, associated with it (stored in the protected tag region). This metadata is maintained automatically by the tagging module as the application is executed. The ratio of tag size to data was chosen because the VPU generates 512 bytes of data for every operation, such as \texttt{add-noise}. Thus, updating a single tag for 512 bytes minimizes the additional memory traffic required for tags while ensuring correctness. We now describe the mechanisms that update the tags. 


\textbf{Tagging module components:}
\label{tagging}To ensure DP guarantees for each training iteration, \X only needs to ensure that the tags indicate whether the data has been properly transformed using the noising operation (\circled{a} in Fig. \ref{fig:key_mod}). However, both instructions (\texttt{add-noise} and \texttt{audit}) should be executed to correctly implement the $noise\_l2$ operation (\circled{a} in Fig. \ref{fig:key_mod}) \ie $\ell_2$-norms computed in each epoch should be less than $C_{th}$. To ensure that this check can be performed by the PMA, we tag the noised data with the value of \emph{epoch} register. When the application attempts to send any noised data out of the device, the PMA ensures that the corresponding tag value is less than the \emph{CStatus} register (if it is not $0$). The \emph{CStatus} register indicates the first epoch in which the checks for \texttt{audit} instruction failed. Thus, only noised data that satisfies the required guarantees can be shared out of the device. 

We add additional SRAM buffers (tag buffers) to store tags alongside data 
. We design the tag buffers of \X such that the accesses for tags do not introduce any contention for the ports of the SRAM buffers which store application data. The on-chip SRAM buffers in ML accelerators are organized as multiple independent banks, and multiple hardware modules, each equipped with dedicated ports~\cite{gemini,genesys}. We implement the tag buffers as an additional bank with its own independent port in each module. For example, in a TPU~\cite{tpuv2} like design, we add two ports. This ensures that there is no port contention for accessing the application data and the corresponding tag in the same cycle. These tag buffers are only added to buffers (\circled{5}) where the results of the operations are written. A tagging module~(\circled{2}) is implemented in the VPU, which appropriately sets the tags of results generated by the instructions. \X implements a memory tagging unit (MTU) (\circled{6}) that moves the tags between on-chip buffers and device memory along with data. The MTU also contains a configuration register (\textit{tag-br}), which stores the base address of tag region and is used to calculate addresses for storage and retrieval of tags. 

\textbf{Tagging Operation:}
In ML workloads, only the \textit{noising} operation can generate data that is safe. Thus, we leverage this to enable highly efficient tagging.  The results from the systolic array are always marked as sensitive. The results generated in the VPU are handled by the tagging module~(\circled{e} in Fig. \ref{fig:proc_changes}). If the \textit{op-code} of the instruction corresponds to the \texttt{add-noise} instruction, the results are tagged with the value of \emph{epoch}, thus being marked as ``noised". Any instruction other than the \texttt{add-noise} sets the tags of the corresponding results to $0$. Thus, \X avoids fetching the tags from the device memory for every operation \ie any data fetched from the device memory is tagged with $0$ in the on-chip buffers. When the data is written back to the device memory, the tags from on-chip buffers are used to update the corresponding tags in the device memory. The tagging module also checks the tags of the operands of control flow instructions. If the tag is $0$, the CStatus register is updated to indicate a failure \ie control flows dependent on sensitive data are not allowed.

However, in a DP-based FL application, the application may execute multiple training iterations in a single round, while noising gradients in each iteration. For such applications, the noised gradients from each training iteration are aggregated at the end of a training round to calculate the round update. \X allows aggregating noised data using the vector addition instruction (\texttt{vadd}) while tagging the result as ``noised". To facilitate this, a custom instruction (\textit{load-tagged}) is added, which fetches the tags along with the data from the off-chip memory into the on-chip buffers. When the application executes a \texttt{vadd} instruction on noised data, the tagging unit checks the tags of the source operands. If all the source operands are ``noised", \ie have a non-zero tag value, then the result is tagged with the tag value which is the higher of the two, or zero otherwise. This satisfies the required guarantees because the result is only tagged as ``noised". The PMA performs the check to ensure that the clipping condition is satisfied for the corresponding epoch before the application is allowed to share aggregated data out of the device.

\subsection{Integration with ML frameworks}
\label{programming}
\X can be seamlessly integrated with ML frameworks via APIs implementing core operations (noising, audit). Aggregators insert these APIs into training code, which compilers translate to \X instructions. Since \X instructions mirror existing ML accelerator instructions (e.g., \texttt{add-noise} parallels \texttt{vector-add}), compiler extensions are straightforward. This approach reduces programming burden compared to prior works requiring custom languages or frameworks~\cite{solo,duet,lightdp,fuzzi,distance,dfuzz}, while offering greater flexibility than libraries with fixed DP mechanisms~\cite{dduo,diffprivlib,ektelo} by enabling custom DP implementations.

\section{\Xe : Extension of \X}

While \X enables FL applications that use per-batch clipping, it does not support per-example clipping. We propose \Xe, an extension for per-example clipping and subsampling (\S \ref{back_dp}) to enable higher accuracy under the same budget. For subsampling, we must ensure records for each training batch are \emph{randomly} sampled from the overall dataset (\S \ref{back_dp}). In this discussion, we assume that one training round consists of one iteration (there is however no restriction). Per-example clipping requires a modified clipping condition described in \S \ref{approach}: \textbf{Clipping condition:} The $\ell_2$-norm of the gradients computed for \emph{each training sample} must not exceed $C_{th}$, as configured by the PMA.  
\begin{figure}[tb]
\begin{center}
\begin{adjustbox}{%
    max totalsize={\linewidth}{\textheight}
}
    \includegraphics[trim={0.2cm 0mm 0 0cm 0mm},clip]{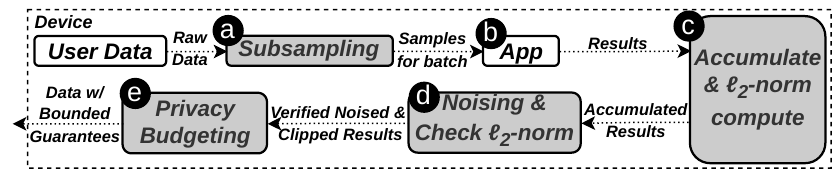}
\end{adjustbox}
    \vspace{-0.5cm}
    \caption{Key operations of \Xe.}
    \label{fig:key_mod_ex}
    \vspace{-1cm}
\end{center}
\end{figure}
\subsection{\Xe: Key Idea}

Fig. \ref{fig:key_mod_ex} shows the four operations implemented by \Xe to provide the necessary guarantees: 1) $sub\-sampling$ (\circled{a}), 2) $acc\_l2$ (\circled{c}), 3) $noise\_chk$ (\circled{d}), and 4) $privacy\ budgeting$ (\circled{e}). 
$Sub\-sampling$ (\circled{a}) ensures that each batch is randomly sampled from the training dataset.  $Acc\_l2$ (\circled{c}) computes the $\ell_2$-norm for each of the per-example gradients. This operation then accumulates the gradients to compute the gradients for the entire batch. $Noise\_chk$ (\circled{d}) noises the gradients and checks that the $\ell_2$-norms of the per-example gradients are less than $C_{th}$. $Privacy\ budgeting$ (\circled{e}) is performed as in \X (\S \ref{budget_cond}). We now explain how these operations check that the necessary conditions are satisfied. We use gradients in our explanations below, but as in \X, these operations can be performed on any data without violating DP guarantees. All data in the device memory is associated with tags. Fig. \ref{fig:mapping} shows the tag fields in \Xe: an identifier of the training record used in the computation ($rid$), whether the data is ``noised" ($s$), and the round in which the data was generated ($depoch$). Only data tagged as ``noised" ($s=1$) and also satisfies the conditions in \S\ref{xe_clip}-\S\ref{xe_budget} can be shared out of the device. 
\subsubsection{Enabling subsampling} 
To ensure that the privacy cost of training is calculated correctly, \Xe must ensure that the application can only access a batch size (\texttt{b}) number of samples every iteration that are randomly sampled from the training dataset. Thus, the training data is stored in a designated region of device memory (record region) initialized by the PMA. \Xe introduces a new instruction (\texttt{load-record}) to allow the application to fetch data from the record region. If the application uses any other instruction to load data from the record region, it is disallowed from sending data out of the device. The \texttt{load-record} instruction takes as input a \emph{record index} (which is between $1$ to \texttt{b}). This index maps to a randomly selected training record as determined by the PMA at the beginning of each training iteration. \Xe tracks the index of the record used to calculate any gradient using IFT. Any record and the results of computation using the record are tagged with its index in the $rid$ field. The $rid$ field is set to $0$ if multiple records are used to compute the gradients. 
\begin{figure}[h]
\begin{center}
\vspace{-0.5cm}
\begin{adjustbox}{%
    max totalsize={0.6\linewidth}{\textheight}
}
    \centering
    \includegraphics[trim={1mm 1mm 1mm 1mm},clip]{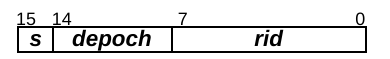}
\end{adjustbox}
    \caption{Tag fields for \Xe.}
    \vspace{-0.75cm}
    \label{fig:mapping}
\end{center}
\end{figure}
\subsubsection{Verifying the clipping condition}
\label{xe_clip}
We introduce a custom instruction (\texttt{acc-grad}) to perform the $acc\_l2$ operation. This instruction is used by the application to accumulate per-example gradients and obtain the gradients for the batch. The results of this instruction are tagged with the value $-1$ in the $rid$ field.  When the instruction is executed, gradient $\ell_2$-norms are calculated for each record (based on the tagged $rid$). The application must call the \texttt{audit} instruction at the end of every training iteration. When the \texttt{audit} instruction is called, each calculated $\ell_2$-norm is checked to be less than $C_{th}$. 
\subsubsection{Verifying the noising condition} 
\label{xe_noise}
The application must use the \texttt{add-noise} instruction to noise the accumulated gradients calculated using \texttt{acc-grad}. \Xe uses the $rid$ field in the tags to ensure that the application is performing the noising operation on the accumulated gradients. Only the results of the \texttt{add-noise} instruction are tagged as ``noised". 



\subsubsection{Verifying the budget condition}
\label{xe_budget}
\Xe tracks the number of iterations executed by the application using a hardware counter (\emph{epoch}) which is incremented when the \texttt{audit} instruction is executed. When the application tries to share data out of the device after $n$ iterations, the PMA computes the privacy cost as $(q\sqrt{n}\epsilon, \delta)$, where $q$ is the sampling probability. The application is allowed to share data out of the device only if the computed privacy cost is less than the budget. 

To ensure that the calculated privacy cost is accurate, \Xe must ensure that the application does not use any data that was computed using records accessed in a previous iteration. To this end, \Xe tags the results of all computations on the accelerator with the value of the \emph{epoch} register in the $depoch$ field. When an application fetches any gradients from memory, the $depoch$ field is checked to ensure that the gradients were generated in the current \emph{epoch} (and thus only using training data from the current iteration).

\subsection{\Xe: Detailed Design}
The device memory in \Xe is divided into 5 segments: the model, records, scratchpad, secure scratch, and tags. The PMA configures the starting addresses of each of these regions using the \emph{record-br}, \emph{scratch-br}, \emph{noise-br}, and \emph{tag-br} registers. The PMA also stores the size of each record in the \emph{rec-sz} register. The model weights and the noised gradients are stored in the model region, while the scratchpad region is used to store any intermediate results, including the gradients. The records region is used to store the sensitive training data. The secure scratch region is used to store the noise samples generated by the PMA and the index to record mappings for each batch. This mapping is generated by the PMA on the host CPU. The secure scratch and tags regions are protected and cannot be accessed by the application directly.

\Xe uses a 2-byte tag for every 512 bytes of data. The data in the model region is considered to be the ``noised" data ($s$ is initialized to $1$) by the PMA. A 1-byte $rid$ field (as in Fig. \ref{fig:mapping}) can support a maximum batch size of 254. The size of the $rid$ field can be increased, thereby increasing the size of the tags, to support larger batch sizes.


\subsubsection{Modifications to the tagging unit} The tagging unit (\circled{e} in Fig. \ref{fig:proc_changes}) is modified to update the tags of the results based on the tags of the source operands: 1) $s$: For any instruction except \texttt{add-noise}, it is set to $0$ if any operands have $s=0$, otherwise it is set to $1$. It is set to $1$ for the results of \texttt{add-noise} always. 2) $rid$: It is set to $-1$ for the \texttt{acc-grad} instruction. For any other instruction, if all the operands with $s=0$ have the same $rid$ value $x$, it is set $x$; otherwise, it is set to $0$. 3) $depoch$: it is set to the current value of \emph{epoch}. 

The tagging module also checks that all the operands of the \texttt{add-noise} instruction have $rid = -1$ and that all the operands of the \texttt{acc-grad} instruction have $rid\neq0$. If either check fails and $CStatus=0$, then $CStatus$ is set to the value of \emph{epoch}. The value of the $CStatus$ register is used by the PMA to ensure that the condition was satisfied in the corresponding epoch for the data, \ie $depoch < CStatus$.


\subsubsection{Modifications to the memory tagging unit} The MTU (\circled{6} in Fig. \ref{fig:overview}) is modified to read/write tags for all data fetched/written to memory. The MTU checks the \emph{depoch} field of any data fetched from scratchpad region and ensures that it is equal to \emph{epoch} (current iteration count). If the check fails and $CStatus=0$, $CStatus$ is updated to the value of \emph{epoch}. We implement \texttt{load-record} in the VPU (\circled{1} in Fig. \ref{fig:overview}). The \texttt{load-record} instruction takes index \emph{idx}, offset (\emph{off}), target \emph{tgt}, and size \emph{bsize} as operands. The mapping of \emph{idx} to an index of a record \emph{tid} in the training dataset is fetched from memory and used to calculate the base address for loading \emph{bsize} bytes from main memory, onto on-chip buffers at the location indicated by \emph{tgt}. The base address is calculated by adding \emph{off} to the product of \emph{tid} and \emph{rec-sz} register.

\subsubsection{Modifications to the noising unit} 
The noising module calculates the $\ell_2$-norm for the gradients computed using each record when the \texttt{acc-grad} instruction is executed. Similar to \X, the partial sum registers ($p_i$, \circled{d} in Fig. \ref{fig:proc_changes}) are updated with the sum of squares of all gradients that have the same $rid$. The partial sums for each record are stored in the secure scratch region of device memory. When the \texttt{audit} instruction is executed, the partial sums for each record are fetched from the device memory to compute the gradient $\ell_2$-norms for each record. Similar to \X, a check is performed to ensure that the $\ell_2$-norm for each record is less than $C_{th}$ and the \emph{CStatus} register is updated if the check fails. The \texttt{add-noise} instruction adds noise to the data as in \X.

\section{Discussion}
\label{discussion}
\label{limitations}

\label{sec_analysis}
\noindent
\textbf{Security Analysis} 
\X and \Xe provide formal guarantees that: (1) any computational results derived from sensitive training data on the accelerator and transmitted from the device satisfy $(\epsilon,\delta)$-DP, and (2) the cumulative privacy expenditure remains within the allocated budget. Under the threat model we describe in \S\ref{assumptions}, we examine four different attack scenarios wherein malicious applications attempt to circumvent privacy protections. 


\noindent
\textit{Attack 1: Sharing Unnoised Data.} A malicious application attempts to transmit any transformation of the training data (including the gradients) without noising them.
\X's tagging module automatically marks the results of all computations performed using the accelerator as sensitive (tag = 0) and this data is not allowed to be shared out of the device. 

\noindent
\textit{Attack 2: Bypassing Gradient Clipping.} The application attempts to share the noised gradients generated using \texttt{add-noise} instruction and executes \texttt{audit} but skips gradient clipping. \X's noising module automatically calculates the $\ell_2$-norm during the \texttt{add-noise} operation and tags the results with the current \emph{epoch}. The $\ell_2$-norm of the gradients being noised exceeds the $C_{th}$ threshold as they are unclipped, the \texttt{audit} instruction updates the $CStatus$ register (tracks the first training iteration where the checks for the audit instruction failed) with the current epoch value. When the application tries to share this noised data, the PMA detects that the tag values are equal to the value of the $CStatus$ register. This violates the requirement that tags be strictly less than the value of the $CStatus$ register, and thus the PMA will terminate the application. The application may also attempt to bypass gradient clipping by not executing the \texttt{audit} instruction. The noised gradients generated using the \texttt{add-noise} instruction are tagged with the current \emph{epoch}. However, the \emph{epoch} is not incremented at the end of training round (\texttt{audit} instruction was not executed). When the application attempts to share these noised gradients out of the device, the tag values will not be less than the current \emph{epoch} and thus the PMA will terminate the application. 

\noindent
\textit{Attack 3: Exceeding Privacy Budget.} The application correctly clips and noises gradients but executes more training iterations than allowed by the privacy budget. Each \texttt{add-noise} operation tags results with the current \emph{epoch} value, which increments with each \texttt{audit} execution. When the application attempts to share final gradients, the PMA calculates total privacy cost from tag values, and the PMA terminates the application as the computed cost will exceed the budget.

\noindent
\textit{\Xe-Specific Attack: Bypassing sampling.} In \Xe, applications may attempt to reuse intermediate results, such as feature vectors, from previous iterations to bypass sampling mechanisms. The intermediate results computed in previous iterations have been tagged with the value of the \emph{epoch} register (which was incremented at the end of the iteration) in the $depoch$ field. When this data is loaded in the subsequent iteration, the Memory Tagging Unit (MTU) compares the $depoch$ field of the tags with the current value of \emph{epoch} register, and the MTU would detect the mismatch in the fields. The $CStatus$ register will be updated with the \emph{epoch} value, which causes the PMA to terminate the application when the application attempts to share this data out of the device. 

\noindent
\textit{Noising and sharing data that is not the gradients.} Our framework provides DP guarantees to arbitrary results derived from sensitive training data. When applications apply \X instructions to transform the training data, the resulting noised data still satisfies DP \ie the sensitivity of the mathematical function representing the transformation is bounded by $C_{th}$, and noise was sampled according to Eq. \ref{equ_ldp}. Thus, sharing this data does not result in any privacy violations.

\noindent \textbf{Formal Proof:} We include a formal proof for \X. Proofs for \Xe, and for privacy budgeting can be constructed analogously. Let $f: \mathcal{D} \rightarrow \mathbb{R}^n$, $n\in \mathbb{Z}^{+}$ be a vector-valued function computed by an application running on an accelerator with \X, where $\mathcal{D}$ is the space of datasets, $\mathbb{R}$ is the set of real numbers and $\mathbb{Z}^{+}$ is the set of positive integers. We show that \emph{regardless of the computation} implemented by the application (\circled{c} in Fig. \ref{fig:key_mod}), \X theoretically guarantees $(\epsilon, \delta)$-DP. For any data identified as safe-to-share, the application has successfully executed $noise\_l2$ (\circled{a} in Fig. \ref{fig:key_mod}) for \X \ie it is always the result of $\mathcal{M}(x)$ as defined below. 
\begin{spacing}{0.9}
{
\vspace{-1mm}
\begin{theorem}
Given that data shared out of the device for \X is a result of $\mathcal{M}: \mathcal{D} \rightarrow \mathbb{R}^n$, then it satisfies $(\epsilon, \delta)$-DP when $\mathcal{M}$ is a standard Gaussian mechanism~\cite{gauss_dp} such that
\begin{equation}
\begin{aligned}
\mathcal{M}(x) = f(x) + \mathcal{N}(0, 4C_{th}^2z^2I_n) \textrm{, } \\ 
\textrm{ where } z = \sqrt{2\log(1.25/\delta)} / \epsilon   .
\end{aligned}
\label{m1}
\end{equation}
\end{theorem}

\begin{proof}
By~\cite{gauss_dp}, $\mathcal{M}$ is differentially private if $\max \|f(D) - f(D')\|_2 = 2C_{th}$ for two adjacent datasets $(D,D')\in \mathcal{D}$.
Recall the $noise\_l2$ operation (\circled{a} in Fig. \ref{fig:key_mod}) in \X ensures $\|f(D)\|_2 \leq C_{th}$. Thus, using the triangle inequality, we can calculate the following bound:
\begin{equation}
\begin{aligned}
    \|f(D) - f(D')\|_2 & \leq \|f(D)\|_2 + \|f(D')\|_2 \\
    &\leq C_{th}+C_{th} \leq 2C_{th}
\end{aligned}
    \label{equ_l2}
\end{equation}
Thus, $\mathcal{M}$ satisfies $(\epsilon, \delta)$-DP, which concludes the proof.
\vspace{-0.5em}
\end{proof}
}
\end{spacing}


\noindent
\textbf{Budget allocation:} Determining appropriate budget allocations is needed for long-term deployments to ensure that the cumulative privacy cost over time provides sufficient protection. \X’s PMA can be integrated with privacy management systems~\cite{cohere,kube} to budget across multiple applications/mechanisms (including potential collusion). When combining different DP techniques for the same data (\eg raw data noising plus DP-FL), and when collusion among multiple aggregators using the same dataset is assumed, the PMA can calculate the privacy cost using composition~\cite{reyni,dp}.

\noindent
\textbf{Generalizability of \X}: 
\X and \Xe together support all LDP FL+DP algorithms that clip gradients by $\ell_2$ norm and add Gaussian noise~\cite{nbafl,fedavg,dpadapt,fl_dp,llm-dp1,llm-dp2,yang2023dynamic,cheng2022differentially}. \X can be trivially extended for LDP schemes which use other noise distributions~\cite{ldp-fl,ldp-fed,shen2023pldpfl}, since they follow the same two steps: clip gradients, then add noise. For example, supporting LDP-Fed~\cite{ldp-fed}, which clips and noises individual gradient elements, only requires modifying the host CPU’s noise-sampling software and the noising unit.

\noindent
\textbf{Limitations:} \this only tracks computation on the ML accelerator, and thus no guarantee is provided regarding any computation performed outside \X (\eg in the CPU). \X is applicable to ML accelerators that are built according to the framework depicted in Fig. \ref{fig:overview}. This includes most systolic array-based accelerators~\cite{tpu,tpuv2,tpuv4,gaudi,genesys} but does not include GPUs where the dataflow and caches cannot be tracked using the \X framework.

\def\com{\texttt{ComB}\xspace}
\def\mem{\texttt{MemB}\xspace}

\section{Methodology}

\textbf{Evaluation Methodology:}
We use a custom cycle-accurate simulator for ML workloads, which is based on the open-source SCALE-sim simulator~\cite{scale_sim} integrated with Ramulator v2.0~\cite{ramulator2}. We model on-chip buffers, double buffering for efficient data transfers, systolic array, and all \X modules. The simulator was validated against Google Cloud TPUv3~\cite{tpuv4} for a wide range of shapes/sizes for GEMM, vector addition, and noising operators (Pearson correlation coefficient \textgreater 0.92). We 
evaluate \X and \Xe with a tag size of 1 and 4 bytes, respectively. \Xe can also be implemented with a tag size of 2 bytes (for batch sizes up to $254$).   
We evaluate \X and \Xe on four baselines that implement DP-enabled ML training but do not provide DP guarantees in hardware: (i) \texttt{TPU}~\cite{tpuv2}: models the Google TPUv3~\cite{tpuv2} with weight-stationary dataflow. (ii) \texttt{DIVA}: that proposes an outer-product based dataflow~\cite{diva} (iii) \texttt{DIVA-PPU}: that also uses a post-processing unit (PPU) with hardware adder trees to accelerate $\ell_2$-norm computation for the per-example gradients~\cite{diva}. The PPU eliminates the need to fetch computed gradients from memory after they have been computed to compute the $\ell_2$-norm. (iv) \texttt{OS}: an accelerator with output-stationary dataflow with same SRAM parameters as TPU for a fair comparison. The systolic array on \texttt{OS} is based on Shidiannao~\cite{shidia}. We evaluate all accelerators using a 128×128 systolic array running at 1 GHz, 24 MB of on-chip SRAM and assume two accelerators on-chip connected to HBMv2 memory, as in TPUv3~\cite{tpuv2} \ie the memory subsystem provides 450 GB/s bandwidth for each accelerator.

\noindent
\textbf{Benchmarks:}
State-of-the-art DP-based FL algorithms~\cite{nbafl,ldp-fed,shuffledp,fl_dp} are evaluated on image recognition tasks. Thus, we evaluate a range of CNN models: VGG16 (\texttt{vgg})~\cite{vgg16}, ResNet50 (\texttt{r50})~\cite{resnet}, ResNet152 (\texttt{r152})~\cite{resnet}, AlexNet (\texttt{alx})~\cite{alex}, Yolov3 (\texttt{ylo})~\cite{yolo},
GoogleNet (\texttt{goo})~\cite{gnet},  SqueezeNet (\texttt{sqz})~\cite{squeeze}, and MobileNetV2 (\texttt{mob})~\cite{mobilenet}. We also evaluate two transformer models~\cite{bert}, BERT-base (\texttt{brtb}) and BERT-large (\texttt{brtl}), for next-sentence prediction tasks. We use ImageNet~\cite{imagenet} dataset for CNN models and BooksCorpus~\cite{bookscorpus} dataset for transformer models. We evaluate \X and \Xe for one local training round consisting of 10 and 1 iteration(s), respectively with a batch size of 8.

\noindent\textbf{RTL implementation:} We implement the \X and \Xe modules 
in Bluespec System Verilog and use openroad 
for 
area and power overheads. With 
CACTI~\cite{cacti} we 
model the area and power overheads of the SRAM buffers used to store the tags in the accelerators.

\section{Evaluation}
\label{evaluation}

\subsection{Performance Overheads}

Fig. \ref{fig:results_dw} and Fig. \ref{fig:results_dwe} show the slowdown of \X and \Xe, for various workloads during local training, over the respective baseline accelerators. 
We observe that \X incurs only 0.09\%, 0.20\%, 0.18\%, and 0.13\% performance slowdown on average for TPU, DIVA, DIVA-PPU, and OS as a result of the additional memory requests required. \Xe incurs on average 0.31\%, 0.5\%, 0.55\%, and 0.37\%  performance slowdown for TPU, DIVA, DIVA-PPU, and OS, respectively. The small performance impact is because \X and \Xe generate little additional memory traffic: $\sim$0.08\% and $\sim$0.61\% for OS and $\sim$0.12\% and $\sim$0.8\% for the others. 



\begin{figure}[tb]
    \centering
    \begin{subfigure}{\linewidth}
        \centering
        \begin{adjustbox}{max totalsize={\linewidth}{\textheight}}
            \includegraphics[clip]{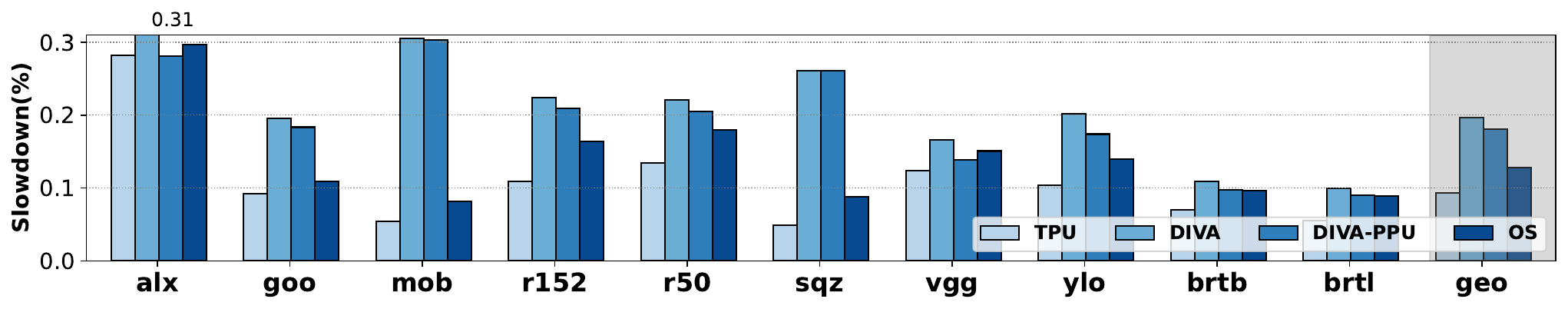}
        \end{adjustbox}
        \vspace{-0.5cm}
        \caption{\X}
        \label{fig:results_dw}
    \end{subfigure}
    \begin{subfigure}{\linewidth}
        \centering
        \begin{adjustbox}{max totalsize={\linewidth}{\textheight}}
            \includegraphics[clip]{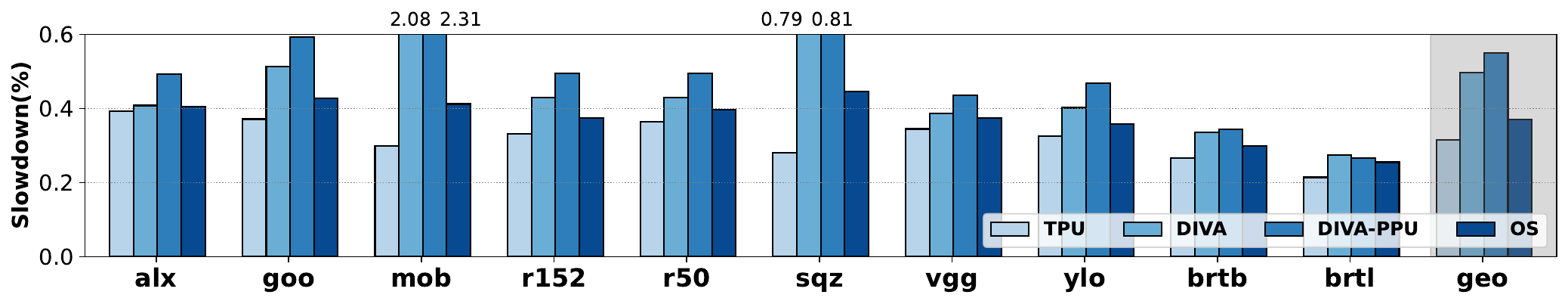}
        \end{adjustbox}
        \vspace{-0.5cm}
        \caption{\Xe}
        \label{fig:results_dwe}
    \end{subfigure}
    \vspace{-0.5cm}
    \caption{Slowdowns for \X and \Xe across various workloads over TPU, DIVA, DIVA-PPU, and OS.}
    \vspace{-0.7cm}
        \label{fig:results_combined}
\end{figure}

Fig. \ref{fig:results_mem_dw} left presents the performance breakdown of \X for the two types of kernels: 1) \com: matrix multiplication kernels in the forward and backward passes, and 2) \mem: kernels for model update and noising. Fig. \ref{fig:results_mem_dw} right shows the breakdown of additional memory requests issued in the two types of kernels. Fig. \ref{fig:results_mem_dwe} shows the breakdown for the two types of kernels (\com, \mem) in terms of execution time (left) and additional memory traffic of \Xe (right) for various baseline accelerators. The per-example gradient computation, clipping, and noising kernels are included in \mem for \Xe. 
We make five observations.

First, the slowdown on TPU is less than the slowdown on DIVA and DIVA-PPU for both \X and \Xe. For example, \X has an average slowdown of 0.09\% and 0.2\%, while \Xe has a slowdown of 0.31\% and 0.55\% on TPU and DIVA. This is because the additional memory requests for \X and \Xe are overlapped with the execution of the \com kernels. The fraction of time spent executing \com kernels is higher on TPU than in DIVA or DIVA-PPU, while the fraction of the additional memory traffic for the \com kernels remains similar. For example, for \texttt{mob} in \X the \com kernels take up 95\% and 80\% of the execution time on TPU and DIVA, while ~97\% of the additional memory requests are issued for \com kernels. Thus, the latency costs of the additional memory operations of \X are hidden better on TPUs, resulting in lower overall slowdowns.
The slowdown of \X and \Xe in OS is less than that in DIVA and DIVA-PPU for the same reason.

Second, the slowdown in OS is more than TPU for both \X and \Xe: \X averages 0.09\% (TPU) vs. 0.13\% (OS), and \Xe 0.31\% vs. 0.37\%. This is due to fewer \com kernel stalls in the OS baseline than in TPU~\cite{diva}, \ie the execution time of \com is lower in OS vs TPU (e.g., r50 on \X: 36\% vs. 41\%). Thus, despite lower additional memory traffic in OS (0.08\% vs. TPU’s 0.12\%), the memory latency for \X/\Xe requests adds additional systolic-array stalls in OS.

\begin{figure}[tb]
    \centering
    \begin{subfigure}{\linewidth}
        \centering
        \begin{adjustbox}{max totalsize={\linewidth}{\textheight}}
            \includegraphics[clip,scale=0.94]{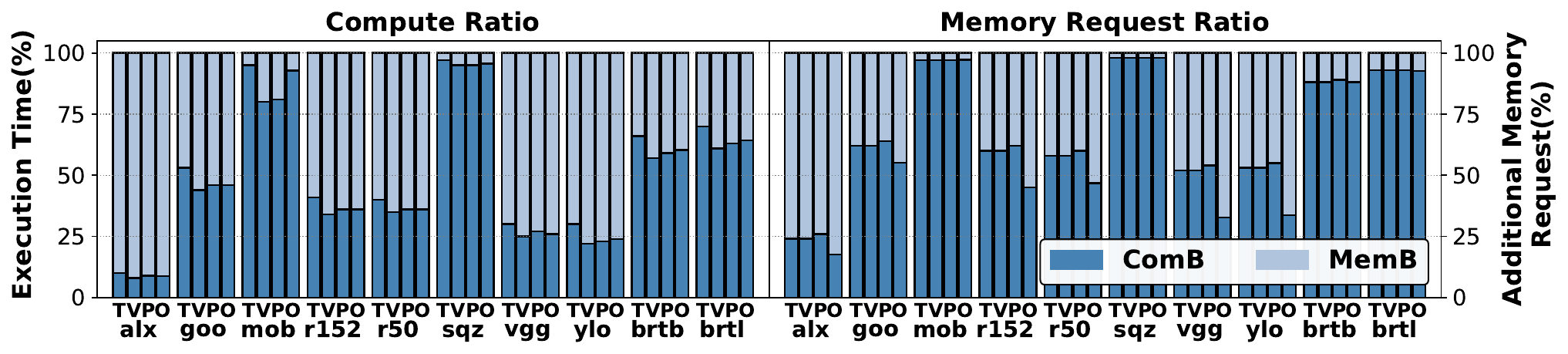}
        \end{adjustbox}
        \caption{\X}
        \label{fig:results_mem_dw}
    \end{subfigure}
    \begin{subfigure}{\linewidth}
        \centering
        \begin{adjustbox}{max totalsize={\linewidth}{\textheight}}
            \includegraphics[]{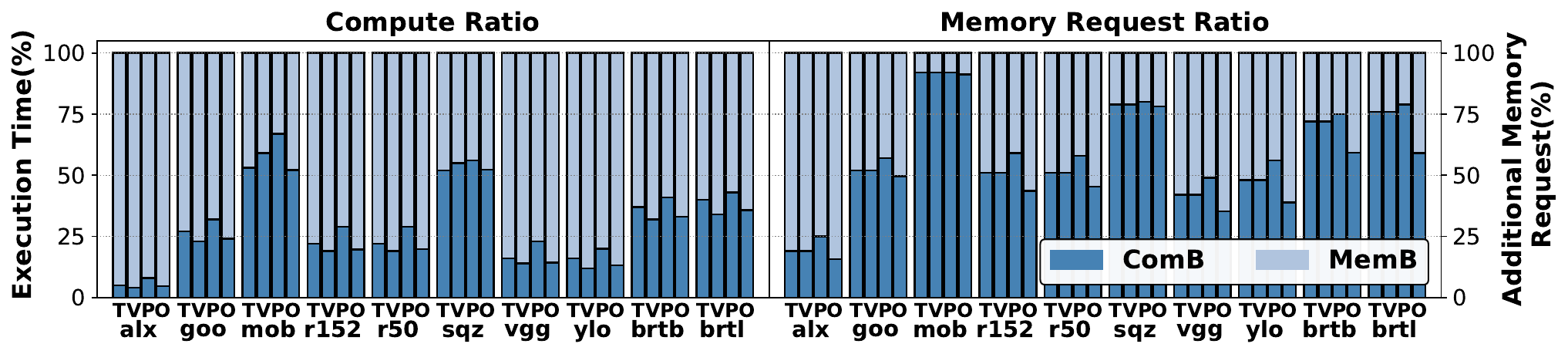}
        \end{adjustbox}
            \caption{\Xe}
        \label{fig:results_mem_dwe}
    \end{subfigure}
        \vspace{-0.5cm}
    \caption{Breakdown of performance and additional memory traffic on TPU(\texttt{T}), DIVA (\texttt{V}), DIVA-PPU(\texttt{P}), \& OS(\texttt{O}).}
    \label{fig:results_mem_combined}
    \vspace{-0.65cm}
\end{figure}

        

Third, \X causes 0.20\% slowdown 
on DIVA, which is 
higher than on DIVA-PPU (0.18\%). The 
difference between DIVA and DIVA-PPU is the presence of the PPU in the latter, which accelerates $\ell_2$-norm computation. However, the acceleration of $\ell_2$-norm computation reduces the portion of execution time spent on \mem kernels, but the ratio of memory requests issued for each type of kernel is similar. 
Additional memory traffic 
of \X for DIVA and DIVA-PPU is also similar (0.12\%). Thus, on DIVA-PPU, a larger fraction of the memory requests of \X are overlapped with systolic array execution, resulting in a lower overall performance impact.

Fourth, \Xe causes a slowdown of 0.5\% on DIVA, while it causes 0.55\%
slowdown on DIVA-PPU. This is because the additional memory traffic of \Xe is higher for DIVA-PPU (0.89\%) in comparison to DIVA (0.77\%). 

Fifth, \X and \Xe cause high slowdowns for \texttt{mob} and  \texttt{sqz} on DIVA, even though the ratio of execution time for \com is high. For example, \Xe causes a slowdown of 2.08\% for \texttt{mob}, which is larger than the slowdowns of other models, even though the \com kernels take 60\% of the execution time. This is because these models have very small sizes for the matrix multiplication kernels, which are efficiently executed on the systolic array. Thus, the additional memory requests of \Xe introduce additional stalls (5.4\% and 1.1\%), resulting in higher slowdowns. 

Overall, we conclude that \X and \Xe are efficient approaches to provide DP guarantees for per-batch clipping and per-example clipping algorithms, respectively.

\subsection{Sensitivity Studies}
\label{sensitivity}
\noindent
\textbf{Tag size:} We present the slowdown of \Xe over the DIVA-PPU baseline for one training iteration when the varying tag sizes in Fig. \ref{fig:results_tag_sz}. The tag size denotes the number of tag bytes per 512 bytes of data, achieved by increasing number of tags, size of each tag, or both.
We observe that the slowdown increases as tag size increases, because the additional memory accesses required to fetch tags increase, causing larger stalls in the accelerator. However, a tag size of four would be the maximum needed for \X and \Xe.

\noindent
\textbf{Batch size:} Fig. \ref{fig:results_bw_sz} shows \X and \Xe performance slowdown  over the baseline for 1 training round 
varying the batch size on DIVA-PPU. 
The slowdown due to \X reduces, as batch size increases, but remains same \Xe. For example, \X incurs a slowdown of 0.18\%, 0.15\%, and 0.13\% for batch sizes 8, 16, and 32, while \Xe incurs 0.55\% for all 
batch sizes. This is because the ratio of the execution time for \com kernels increases for \X, while it remains the same for \Xe, as batch sizes increase. The latency needed to perform the memory operations of \X is overlapped with the 
\com kernels execution, resulting in a lower overall impact on slowdown. 

\begin{figure}[tb]
\begin{adjustbox}{max totalsize={\linewidth}{\textheight}}
    \centering
    \includegraphics[trim={2mm 2.5mm 2mm 2.5mm},clip]{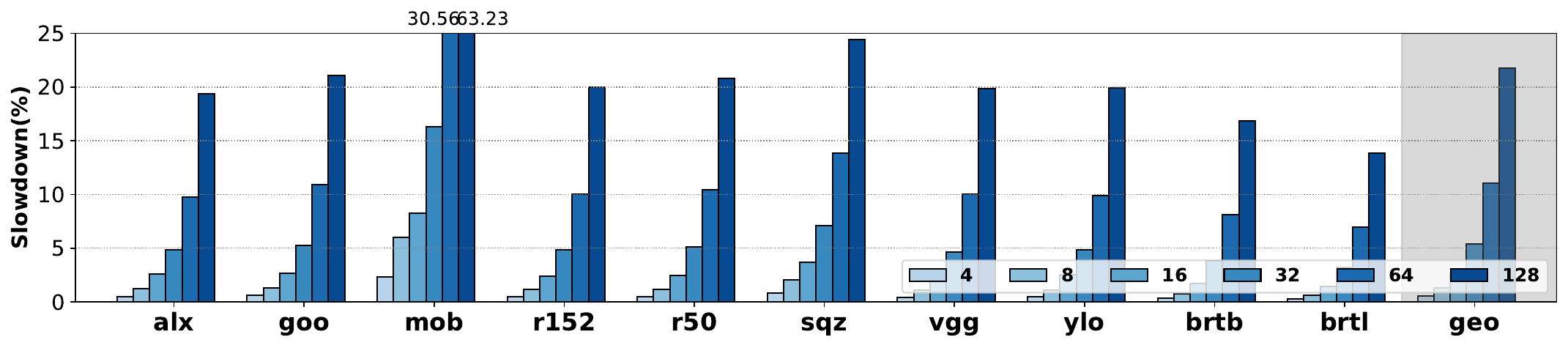}
\end{adjustbox}
    \vspace{-0.6cm}
    \caption{\Xe slowdown across tag sizes on DIVA-PPU.}
    \label{fig:results_tag_sz}
    \vspace{-0.5cm}
\end{figure}

\begin{figure}[tb]
\begin{adjustbox}{max totalsize={\linewidth}{\textheight}}
    \centering
    \includegraphics[]{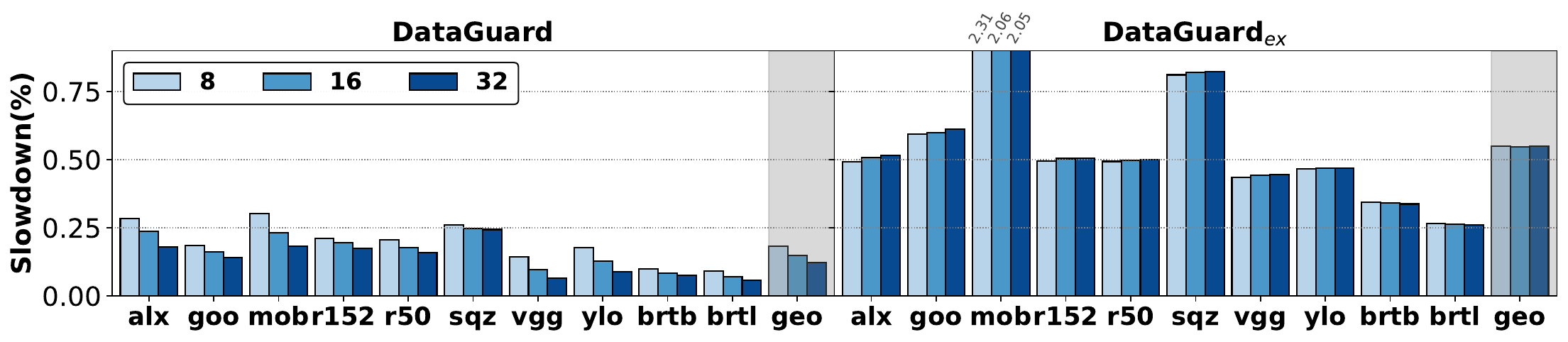}
\end{adjustbox}
\vspace{-0.5cm}
    \caption{Slowdown for various batch sizes on DIVA-PPU.}
    \vspace{-0.7cm}
    \label{fig:results_bw_sz}
\end{figure}
\vspace{-5pt}

\begin{table}[t]
\centering
\caption{ISA additions/semantic extensions introduced by \X/\Xe.}
\label{tab:dwatch_instructions}
\small
\setlength{\tabcolsep}{3pt}
\renewcommand{\arraystretch}{1.05}
\begin{adjustbox}{max totalsize={\linewidth}{\textheight}}
\begin{tabular}{@{}l l p{0.46\columnwidth}@{}}
\hline
\textbf{Instruction} & \textbf{Applies to} & \textbf{Operation} \\
\hline
\texttt{add-noise}        & \X, \Xe & Add Gaussian noise to operands. \\
\texttt{audit}            & \X, \Xe & Verify clipping condition. \\
\texttt{load-tagged}      & \X      & Load data from memory along with tags. \\
\texttt{vadd}  & \X      & Vector add (with tag propagation). \\
\texttt{load-record}      & \Xe     & Load a subsampled training record. \\
\texttt{acc-grad}         & \Xe     & Add two per-example gradients. \\
\hline
\end{tabular}
\end{adjustbox}
\vspace{-0.65cm}
\end{table}

\subsection{Hardware Overheads}
The main overheads are the additional SRAM buffers that store tags and the logic to implement the noising unit. For a TPUv3, whose chip-wide area is 
$\sim$650 $mm^2$ (on 12 nm) and each chip contains two accelerators~\cite{tpuv2},  we add a total of 32KB and 128KB SRAM buffers 
per accelerator for \X and \Xe, respectively. We incur an additional chip-wide area overhead of 0.01\% and 0.05\% (22 nm in CACTI) for tag storage in for \X and \Xe, respectively. The noising unit 
has an area of 0.003 $mm^2$ (
using 7 nm technology~\cite{asap7} targeting a frequency of 1 GHz) per accelerator and costs only an 
extra 0.0008\%  area overhead chip-wide. 

\X and \Xe introduce a set of new accelerator instructions (listed in Table \ref{tab:dwatch_instructions}). Although they require changes to the vector processor, these changes remain intentionally localized and lightweight. The new opcodes extend the vector-processor ISA (\eg \texttt{add-noise} mirrors a standard vector add) and require only modest datapath modifications, namely, dedicated noising and tagging modules. All other changes are confined to metadata handling: small SRAMs for tags, and an MTU that moves tag bytes along with normal data transfers. Consequently, the chip-wide area overhead is merely $\sim$ 0.01\% and $\sim$0.05\% for \X and \Xe, respectively.

\X and \Xe also incur a memory overhead of 64 MB and 256 MB to store the tags for TPUv3, which has a total memory capacity of 32GB~\cite{tpuv2}. The energy overhead of performing tag fetches/writes is $0.08\,\%$ for \X and $0.61\,\%$ for \Xe on the OS accelerator, and $0.12\,\%$ for \X and $0.80\,\%$ for \Xe on the other accelerators. These values were estimated using memory traffic from simulation, assuming energy access for HBMv2 is $3.6\,\mathrm{pJ/bit}$~\cite{energy}.

The noising module adds 0.12 Watts per accelerator causing a chip-wide increase in power consumption of 0.05\% when running at 1 GHz, \ie TDP of TPUv3 is 450 Watt~\cite{tpuv2}. We estimate the static and dynamic power of the additional SRAM modules. 
The static power, obtained from CACTI, is 0.02 Watts and 0.04 Watts per accelerator for \X and \Xe, respectively. The dynamic power was estimated to be 0.0105 Watts and 0.1024 Watts per accelerator for \X and \Xe, respectively. Thus, the total power overhead for \X and \Xe is 0.07\% and 0.1\%.

\vspace{-4pt}
\section{Prior Work}
\label{prior_work}
To our knowledge, \this is the first 
hardware-based framework to provide DP guarantees during ML training on accelerators. We briefly discuss prior work.

\noindent
\textbf{Privacy Guarantees in Hardware:}
Prior works~\cite{simha_cal, serandip, ulp_ldp} propose approaches to noise raw data before the application is granted access to it by modifying the noising mechanisms~\cite{serandip} or designing dedicated hardware modules~\cite{simha_cal,ulp_ldp} for noising. These approaches are limited to certain kinds of sensor data (\eg accelerometers) and were not designed for text, images, or videos. Directly noising training data can also lead to significant accuracy loss in ML training~\cite{dpfyml}. 

\noindent
\textbf{Accelerators for DP-enabled Training:}
Prior works propose accelerators for DP-based ML: DIVA~\cite{diva} is a novel dataflow architecture with a reduction module that accelerates the per-example gradient computations. DINAR~\cite{dinar} proposes a novel noise generation algorithm that relies on pre-generated seeds to generate noise. DPAcc~\cite{dpacc} offloads DP operations to an FPGA placed in the memory path, accelerating DP training by performing these operations in-line during data transfers. These works are orthogonal; they improve performance but require trusting a third-party application for DP guarantees. 

\noindent
\textbf{Software Mechanisms for Protecting User Privacy:}
Prior  works~\cite{taintart,taintdroid} proposed software IFT techniques to detect leakage of sensitive information. These works rely on a privileged software layer running between the application and hardware to instrument tracking and detection. However, these works are not designed for ML applications where gradients computed from sensitive data are sent out of the device. These techniques also rely on a trusted software layer, which increases the attack surface of the system; a software bug can allow malicious applications to leak data. 
Such techniques also incur significant performance overheads (\eg up to 15\%~\cite{taintart}) due to instruction execution overhead for IFT, and additional memory traffic for tag metadata. \X is designed to minimize these overheads for DP. Several prior works~\cite{pinq,gupt,fuzzi,diffprivlib,lightdp,dfuzz,dduo,chorus,solo,ektelo,duet,hoare} propose mechanisms that automatically analyze a program and ensure that it is DP. These works are implemented as libraries~\cite{diffprivlib,ektelo,dduo} or define custom programming languages~\cite{solo,hoare,lightdp}. However, these works only provide guarantees if the program uses specific libraries. It is challenging to ensure that an application is correctly using them without manual analysis by an expert. 

\noindent
\textbf{Ensuring privacy with an untrusted aggregator:} Prior work~\cite{takagi2025securingprivatefederatedlearning} proposes using trusted execution environments (TEEs) on the aggregation server to enforce a differentially private aggregation protocol and enable secure aggregation. In this design, aggregation functions execute inside enclaves, and clients verify the integrity of the server’s code via remote attestation. However, this approach assumes an adversary other than the aggregator or the client. It does not consider the case where the aggregator itself is malicious and controls both the aggregation server and the training applications running on client devices, as in real-world FL deployments~\cite{google_fl}. In such settings, a server-side TEE solution fails to protect client privacy: a malicious aggregator can simply modify the training application to bypass remote attestation checks, violating DP guarantees. In contrast, our threat model assumes no trust in the application running on the client devices.
\section{Conclusion}


We propose \this and \Xe, the first work that enables DP guarantees in hardware for FL applications on ML accelerators. With \X, any data that is shared out of the device satisfies DP guarantees (defined by data owner) regardless of the operations performed by the application. \X requires minimal changes to the accelerator and incurs small performance and area overheads.



\section*{AI Use Disclaimer}
ChatGPT and Claude were used in the preparation of this manuscript for editing and grammar checking. No passages were copied without full author review, and all substantive ideas, analyses, and conclusions are the product and responsibility of the authors.


\bibliographystyle{IEEEtranS}
\bibliography{bibliography}
\widowpenalty=1000
\clearpage

\end{document}